\documentclass[lettersize,journal]{IEEEtran}
\usepackage{amsmath,amsfonts}
\usepackage{array}
\usepackage[caption=false,font=normalsize,labelfont=sf,textfont=sf]{subfig}
\usepackage{textcomp}
\usepackage{stfloats}
\usepackage{url}
\usepackage{verbatim}
\usepackage{graphicx}
\usepackage{cite}
\usepackage{mathtools} 
\usepackage{amssymb}  

\usepackage{amsthm}
\usepackage{algpseudocode}
\usepackage{xcolor}
\usepackage{accents}
\usepackage{mathrsfs}
\usepackage[normalem]{ulem}
\usepackage{cleveref}

\newtheorem{theorem}{Theorem}
\newtheorem{lemma}{Lemma}

\newtheorem{definition}{Definition}
\newtheorem{remark}{Remark}
\newtheorem{assumption}{Assumption}

\newcommand{\red}[1]{\textcolor{black}{#1
}}

\renewcommand{\emptyset}{\varnothing}
\renewcommand{\baselinestretch}{0.98}

\def\mf{\mathbf}

\def\mc{\mathcal}

\def\beq{\begin{equation*}}
\def\eeq{\end{equation*}}
\def\bql{\begin{equation}}
\def\eql{\end{equation}}
\def\bqn{\begin{eqnarray*}}
\def\eqn{\end{eqnarray*}}
\def\bnl{\begin{eqnarray}}
\def\enl{\end{eqnarray}}

\DeclareMathOperator{\sgn}{sgn} 

\newcommand{\att}[1]{$A_{#1}$}

\def\nufast{\nu_\text{fast}}
\def\nuslow{\nu_\text{slow}}
\def\nutrue{\nu}
\def\speed{v}

\def\x{\mf{x}}
\newcommand{\pos}[1]{\mf{x}_{#1}}
\newcommand{\vel}[1]{\mf{v}_{#1}}
\def\rai{r_{A_i}}

\def\widthone{w_\text{1v1}}
\def\dang{\theta_{T}}

\def\reachone{\mc{R}_\text{1v1}}
\def\reachtwo{\mc{R}_\text{2v1}}

\def\intersectone{\mc{I}_\text{1v1}} 
\def\intersecttwo{\mc{I}_\text{2v1}} 
\def\unionone{\mc{U}_\text{1v1}} 
\def\uniontwo{\mc{U}_\text{2v1}} 

\newcommand{\region}[1]{\Theta_{#1}}
\def\cone{\theta_{\mathcal{I}_1}}
\def\wone{w_{\mathcal{I}_1}}
\def\dcone{\dot{\theta}_{\mathcal{I}_1}}
\def\dwone{\dot{w}_{\mathcal{I}_1}}

\def\AttackerStrategy{SS-FH strategy}

\def\vone{V_\text{1v1}}
\def\vtwo{V_\text{2v1}}
\def\Runnercap{\tilde{\theta}}

\begin{document}

\title{
Deception in Differential Games: Information Limiting Strategy to Induce Dilemma
}

\author{Daigo Shishika$^1$, Alexander Von Moll$^2$, Dipankar Maity$^3$, Michael Dorothy$^4$
\thanks{
We gratefully acknowledge the support of DEVCOM ARL grant ARL DCIST CRA W911NF-17-2-0181. The views expressed in this paper are those of the authors and do not reflect the official policy or position of the United States Government, Department of Defense, or its components.
}
\thanks{
D. Shishika is with the Department of Mechanical Engineering, George Mason University, Fairfax, VA, 22030, USA. 
Email: 		{ dshishik@gmu.edu}
}
\thanks{
A.~Von Moll is with the Control Science Center,  Air Force Research Laboratory, WPAFB, OH, 45433, USA. 
Email: { alexander.von\_moll@us.af.mil}
}
\thanks{D. Maity is with the Department of Electrical and Computer Engineering, University of North Carolina at Charlotte,  NC, 28223, USA.
Email: 		{{dmaity@charlotte.edu}}
}
\thanks{
M. Dorothy is with the Army Research Directorate, DEVCOM Army Research Laboratory, APG, MD, 20783, USA. 
Email:  {michael.r.dorothy.civ@army.mil}
 }        
}



\maketitle

\begin{abstract}
Can deception exist in differential games?
We provide a case study for a Turret-Attacker differential game, where two Attackers seek to score points by reaching a target region while a Turret tries to minimize the score by aligning itself with the Attackers before they reach the target.
In contrast to the original problem solved with complete information, we assume that the Turret only has partial information about the maximum speed of the Attackers.
We investigate whether there is any incentive for the Attackers to move slower than their maximum speed in order to ``deceive'' the Turret into taking suboptimal actions.
We first describe the existence of a dilemma that the Turret may face.
Then we derive a set of initial conditions from which the Attackers can force the Turret into a situation where it must take a guess.
\end{abstract}


\section{Introduction}

\IEEEPARstart{D}{ifferential} games formulated for agents with first-order dynamics leads to ``simple'' deterministic control as the optimal strategies \cite{isaacs1965differential, shishika2018local-game, dorothy2021one, yan2023multiplayer}.
In many cases, the agent heads straight towards the optimal point at its maximum speed, which stems from the Mayer-type payoff functional used to model the agent incentive \cite{vonmoll2022circular}.
Any deviation from this strategy generates suboptimality in terms of the time or the distance metric used to define the payoff function.

Aside from the payoff structure, one of the main reasons behind this common result is the assumption of the perfect and complete information: i.e., knowledge about the opponent's state (position), its intention (payoff), and capability (speed).
Due to such strong assumptions on the informational structure, the optimal or equilibrium strategies are often clearly defined based on the instantaneous snapshot of the scenario, which has no dependencies on the history of the states.
In other words, a player cannot use its action to influence the decision of its opponent who is employing an equilibrium strategy.

Concerning uncertainty in capability, reference~\cite{nath2022two-phase} analyzed a scenario in which an Evader wishes to delay capture by a faster Pursuer for as long as possible but does not know the Pursuer's (finite) turn radius.
A heuristic evasion policy is presented therein which makes no assumption about the Pursuer's turning capability and only uses available information (relative position, maximum speeds).
Since the Evader is not actively trying to estimate the Pursuer's turning radius there is no opportunity for the Pursuer to influence (i.e., deceive) the Evader.

Regarding uncertainty in the state of the system, this could be classified under the umbrella of stochastic differential games wherein one is concerned with \textit{expected} payoff.
References~\cite{yavin1986tracking,yavin1987pursuit-evasion} contain a particularly relevant example in which a Camera is tasked with tracking a mobile agent, the Evader.
The Evader has some control over its motion but there is also a stochastic process added to its dynamics (which could represent, for example, wind).
Generally, in stochastic differential games, the information structure is still symmetric, meaning both agents have the same uncertainty over the dynamics of the system.

Whenever agents within a system do not all possess the same information (i.e., there is some information asymmetry) signaling becomes a critical consideration. 
%
Deception may be classified as a type of signaling in which the information transmitted to other agents is not entirely truthful or the intended perception is not an accurate representation of reality.
The idea of deception is important in both civilian and military security applications.
It has been studied extensively in cyber-physical systems (CPS)
involving Attacker and Defender~\cite{etesami2019dynamic,sayin2021deception-as-defense}.
%
For scenarios involving mobile agents, the agents' control, and its effects on the state of the system, adds to the complexity of the type and nature of information that is made available to an adversary~\cite{rubinovich2020alternate,yavin1987pursuit-evasion}.
In~\cite{yavin1987pursuit-evasion}, the authors consider a pursuit-evasion problem wherein the Evader, in addition to having control over its dynamics is able to affect to the Pursuer's perception of the system state in various ways by introducing noise or false signals.
This may be seen as a deception by means of a sensor attack.

We are interested in investigating the idea of \emph{deception by motion}, where the information leaked from one player is tightly coupled to its behavior in the physical world.
More specifically, for motion control problems, the agent's intention as well as its motion capability are tightly coupled with its trajectory.
This coupling leads us to study the trade-off between the \emph{positional advantage} and the \emph{informational advantage}.

To explore the types and nature of deception in differential games, or even its existence, we consider a variant of the Turret-Attacker differential game~\cite{vonmoll2022turret-runner-penetrator}, in which there is an informational asymmetry.
Specifically, we consider a scenario where the Turret does not have precise information about the Attackers' maximum speed.
As is the case with many other pursuit-evasion games, relative maximum speed information is critical in selecting the optimal strategy.
Therefore, there may be an incentive for the Attacker to move slowly and hide its true capability to maintain \emph{informational advantage}.
On the other hand, the literature in pursuit-evasion games have shown that it is almost always the case that the players should move at the maximum speed to retain \emph{positional advantage}.
This paper investigates the interaction between these potentially conflicting goals, and demonstrate that it is indeed possible to have a case where the Attacker benefits from moving at the speed slower than its maximum.

The contributions of this paper are:
(i) demonstration of a scenario where the Attacker benefits from playing a deceptive / information-limiting strategy;
(ii) characterization of the initial conditions that can lead to deception (i.e., the Defender facing dilemma); and 
(iii) derivation of the strategies for the Defender (resp.~Attacker) to avoid (resp.~induce) dilemma.
We believe that the notion of ``information-limiting'' strategy extends to a broader class of asymmetric information games beyond those related to unknown speeds (to be formalized in Sec.~\ref{sec:information_limiting}).
Our work serves as a starting point to identify conditions that lead to the existence of deception by motion.


\subsection{Related Work}

\paragraph*{Signaling through Action}
A number of existing works have considered situations where one's actions become a signal for his opponent, within traditional discrete game settings.
In the setting of Colonel Blotto games,~\cite{fuchs2012sequential} described a game where Player B was equipped with a sensor network, so that Player A must consider how its strategy is going to signal B about its strategy.
%
A similar setup was described in~\cite{hespanha2000deception} wherein, prior to the playout of a matrix game, a Defender could ``show'' the location of any of its defensive assets to the Attacker equipped with an imperfect detector.
Depending on the reliability of the Attacker's detector, the Defender could nullify the benefit of the Attacker's pregame observations via a deceptive signalling strategy.
In context of repeated matrix games,~\cite{israeli1999sowing} showed that Player A can sometimes increase its minimum payoff by occasionally playing suboptimally in an effort to convince Player B that Player A's objective is something different.

\paragraph*{Deception by Motion}
A number of existing works have also considered deception in the context of motion control.
A typical formulation optimizes the path of a moving agent to reach its goal while minimizing the quality of the inference an observer is trying to make about the location of the agent's goal~\cite{Ornik2018DeceptionIO,Dragan2015DeceptiveRM,Topcu2022resal,Topcu2022supcontrol}. The moving agent must consider the trade-off between positional and informational advantage in terms of reaching its true goal quickly versus making it less obvious. Specifically, \cite{Dragan2015DeceptiveRM} and \cite{Topcu2022resal} formalized the notion of \emph{ambiguity} (hiding information about true goal), and \emph{exaggeration} (moving towards a decoy goal to send a false signal) as two ways to measure deceptiveness.

A common assumption made in these works is that the observing agent uses a \emph{prescribed} estimator/inference policy, and the deceiving agent leverages the knowledge of its structure. The deceived agent is also often so naive that it is not aware of the possibility of being deceived. Furthermore, since no decision is made by the observing agent on its action or inference policy, such formulation normally boils down to a one-sided optimization. 
The works on goal recognition by a passive observer also fall into this category~\cite{ Ramrez2010ProbabilisticPR,Masters2017DeceptiveP,Masters2019CostBasedGR,Wayllace2016GoalRD,Kulkarni2018ResourceBS,Masters2018CostBasedGR, Pereira2017LandmarkBasedHF}. While ~\cite{comandur2023desensitization} did not specify an estimator/inference policy for each agent, it similarly prescribed a desensitized, multi-objective payoff for each agent, not allowing either to reason directly about how its opponent may be acting deceptively.

We are interested in studying the possibility of deception without making the assumptions discussed above. We do not prescribe an estimator for the observing player, but instead allow it to select its own policy, which makes the problem a two-player game. Since the observing agent must choose its action based on the observed motion of the deceiver, the success of deception can be directly measured through the influence on decisions of the observer, and ultimately by the outcome of the game.
While existing works use observer's belief in the objective function (making deception itself to be the goal \cite{Ornik2018DeceptionIO, Dragan2015DeceptiveRM, Topcu2022supcontrol,Masters2017DeceptiveP}), our formulation views deception as a tool to accomplish underlying mission objectives (i.e., improve the game outcome). 




\section{Preliminaries}
This section formulates the two-Attacker one-Turret differential game with uncertainty in the Attackers' capability. 
We then present an illustrative example where the Turret may face a dilemma.
Finally, we define the notion of information-limiting strategy, which we use to discuss deception in this work.

\subsection{Problem Formulation}\label{sec:problemFormulation}

The game is played between a Turret (generally notated with subscript $T$) and two Attackers (\att{1} and \att{2}).
The Attackers seek to reach the target region defined by a unit circle centered at the origin.

The Turret is placed at the origin and its pointing angle is denoted by $\dang\in[-\pi,\pi]$.
The Turret has first-order dynamics:
\begin{equation}
    \dot{\theta}_T = \omega_T,
\end{equation}
where $\omega_T\in[-1,1]$ is the control input.
The Attackers also have first-order dynamics but in $\mathbb{R}^2$: 
\begin{equation}
    \dot{\x}_{A_i}=\vel{A_i},
\end{equation}
where the maximum speed is given by the speed-ratio parameter $\nu\in(0,1]$, i.e., $\speed_{A_i} \triangleq \|\vel{A_i}\| \leq \nu \leq 1$.
\footnote{
    If $\nu > 1$, then an Attacker will have an angular rate advantage within $r \leq \nu$.
    Therefore, it must only reach this radius in order to guarantee a successful breach.
    If we then treat this radius as the terminal surface and rescale the problem, all of the same analysis will apply.
}
Note that the two Attackers have the same maximum speed~$\nu$.
We also use polar coordinates $[r_{A_i}, \theta_{A_i}]^\top$ to describe the position of \att{i}.
The heading $\phi_{A_i}$ defined for an Attacker with respect to its radial vector (see Fig.~\ref{fig:1v1TurretWinningRegion}).

An Attacker is removed from the game when either of the following happens: 
\begin{itemize}
    \item \emph{Capture}: $\theta_{A_i}(t_{F_i})-\dang(t_{F_i})=\nobreak 0$, or 
    \item \emph{Breach}: $r_{A_i}(t_{F_i})=1$,
\end{itemize}
where $t_{F_i}$ for $i\in\{1,2\}$ denotes the terminal time for the $i$-th Attacker.
The game terminates at $t_F \triangleq \max_{i\in\{1,2\}} t_{F_i}$ when both Attackers are removed.
We consider a zero-sum Game of Kind\footnote{Game with a finite number of possible outcomes; see~\cite{isaacs1965differential}.} in which the terminal payoff is the number of breaches:
\begin{eqnarray}
J(\x_0;\gamma_A,\gamma_T) = \left\{\begin{array}{cl}
0 & \text{if both Attackers are captured} \\
1 & \text{if only one Attacker breaches} \\
2 & \text{if both Attackers breach},
\end{array}
\right.
\end{eqnarray}
where $\x_0=[\x_{A_1}(0)^\top,\x_{A_2}(0)^\top,\theta_T(0)]^\top$ is the initial condition, $\gamma_T$ is the Turret strategy, and $\gamma_A$ is the Attackers' strategy.
The Turret is the minimizer while the Attacker pair is the maximizer.

We consider $\gamma_T$ and $\gamma_A$ to be state-feedback control policies,
assuming that the players have access to the instantaneous state vector, $\x(t)$.
The Attackers have access to the true speed ratio parameter, $\nutrue$, however the Turret only has partial information about $\nutrue$.
We make the following assumptions on the information that the Turret has about $\nutrue$.
\begin{assumption}\label{assm:two_options}
    There are only two possible speed ratios: $\nufast$ or $\nuslow$, where $\nuslow < \nufast$.
\end{assumption}
\begin{assumption}\label{assm:speed_history}
    The Turret has access to the history of the Attackers' speeds.
\end{assumption} 
Based on the two assumptions and the restriction that the two Attackers have the same capability, if the Turret observes any speed $\speed_{A_i}(t)>\nuslow$ at any point in time, then it knows that the true speed ratio is $\nutrue=\nufast$.

\subsection{Motivating Example: Open-loop Strategy}

Consider open-loop strategies where the Turret commits to pursuing the attackers in a particular order, i.e., the Turret commits to going clockwise (CW) or counterclockwise (CCW) until the corresponding attacker is removed from the game.
We first introduce how the Turret's decision and the speed ratio may affect the outcome $J$ even if the game starts from the same initial condition.

\begin{figure}[t]
    \centering
    \includegraphics[trim={3cm 0 0 0},clip,width=0.4\textwidth]{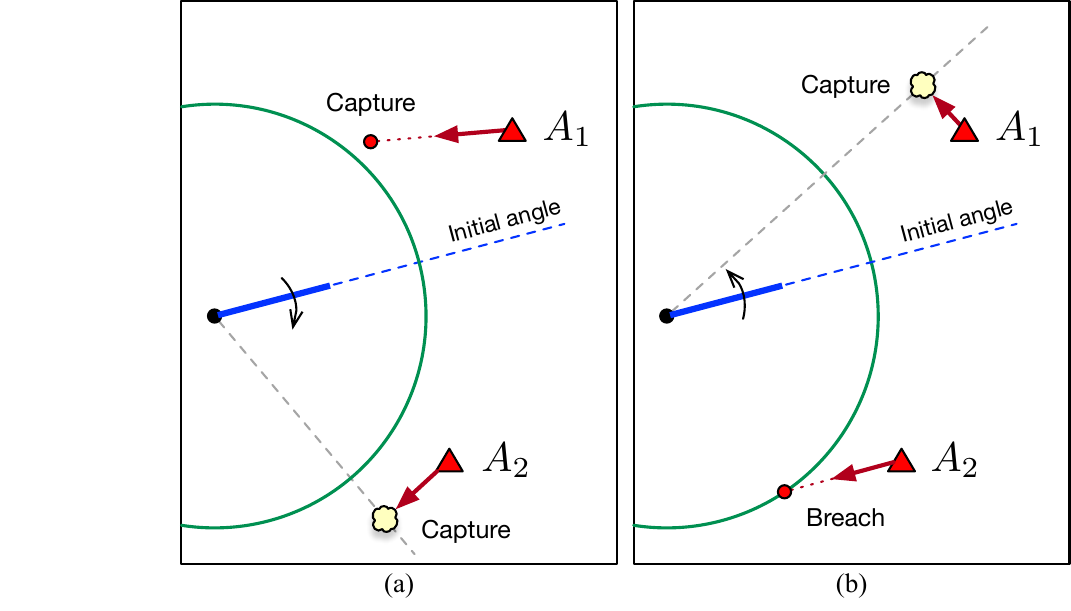}
    \caption{Slow Attackers: one breach ($J=1$) vs zero breach ($J=0$).
    (a)~If the Turret turns CW to capture \att{2} first, then it has sufficient time to capture \att{1} too.
    (b)~If the Turret turns CCW to capture \att{1} first, then \att{2} will be able to breach. 
    }
    \label{fig:motivating_example}
\end{figure}
{
For the state depicted in Fig.~\ref{fig:motivating_example} with a certain speed ratio, $\nuslow$, the Turret will achieve two captures if it goes CW first to capture \att{1} and then turns back to capture \att{2}. 
If it instead goes CCW first, \att{2} will successfully breach.
This example highlights how the decision made by the Turret determines whether the payoff is $J=0$ or $J=1$. 
The conditions that lead to such a scenario and the associated equilibrium strategies will be presented in Sec.~\ref{sec:complete_info}, but for now, the existence of such a configuration is enough to proceed with our discussion.
}

{
Now, from the same initial condition, suppose the speed ratio was some higher value, $\nufast$.
Figure~\ref{fig:motivating_example_fast} illustrates how the payoff is now either $J=1$ or $J=2$, depending on the decision made by the Turret.
What is important to note here is that the optimal direction (i.e., optimal sequence of capture) for the Turret to secure one capture is CCW (i.e., to capture \att{1} first), but this is the opposite of what the Turret should do in the earlier example with $\nu=\nuslow$.
}
\begin{figure}[t]
    \centering
    \includegraphics[trim={3cm 0 0 0},clip,width=0.4\textwidth]{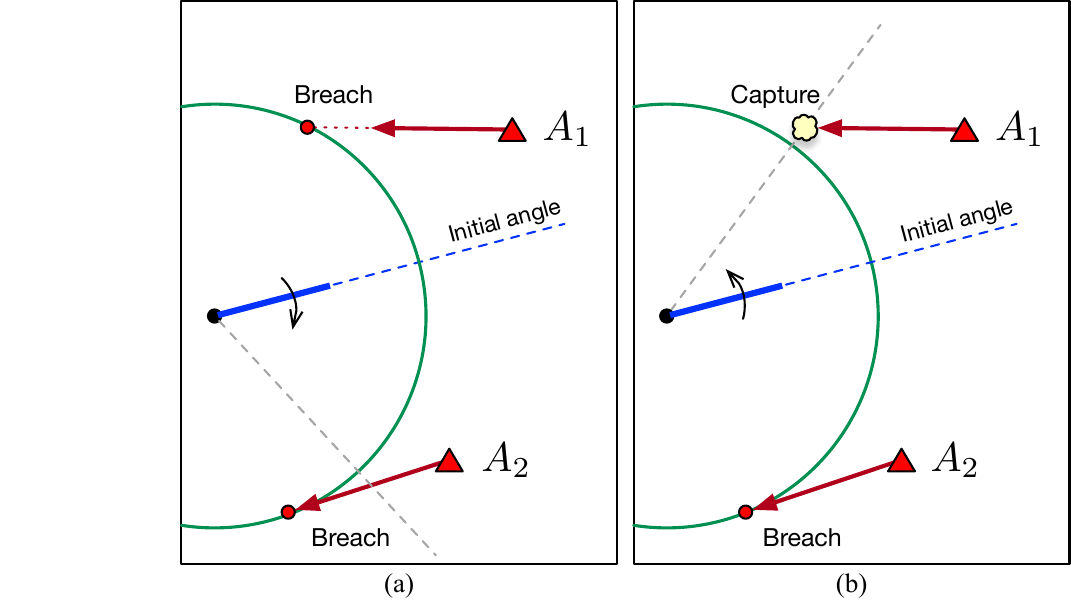}
    \caption{Fast Attackers: two breaches ($J=2$) vs one breach ($J=1$).
    (a)~Attacker \att{2} cannot be captured even if the Turret turns CW. If the Turret wastes too much time going CW (in hope of capturing \att{2}), then \att{1} will also breach. 
    (b)~If the Turret turns CCW, at least \att{1} can be captured.
    }
    \label{fig:motivating_example_fast}
\end{figure}




Finally, consider the case where the Turret is unaware of the true speed of the Attackers: it is either $\nuslow$ or $\nufast$.
If the game was played in an open-loop fashion (i.e., the decision is made only once at the beginning of the game), the Turret faces a dilemma between the following two options:
\begin{itemize}
    \item Play ``aggressively'' by assuming $\nuslow$ and attempt two captures to achieve $J=0$ (at the risk of allowing $J=2$ if the true speed was $\nufast$); or
    \item Play ``conservatively'' by assuming $\nufast$ and ensuring one capture to achieve $J\leq 1$ (at the risk of missing out on two captures if the true speed was $\nuslow$).
\end{itemize}
These cases are summarized in Table~\ref{tab:openLoopMatrix}.
\begin{table}[h]
\caption{A static analysis of the game when the Turret makes a decision in open-loop manner.}
\label{tab:openLoopMatrix}
\begin{center}
\begin{tabular}{|l|l|l|}
\hline
{} & $\nuslow$ & $\nufast$ \\ \hline
CW (aggressive)  & $J=0$      & $J=2$      \\ \hline
CCW (conservative) & $J=1$      & $J=1$      \\ \hline
\end{tabular}    
\end{center}
\end{table}

One may be tempted to analyze the game as a static matrix game\footnote{A player's preference over the above two options is application specific, and therefore, it is not the main interest of our work.} on Table~\ref{tab:openLoopMatrix}, but that is valid only when the game is played in an open-loop fashion: i.e., the Turret decides on the sequence of captures based only on the initial condition.
Our main focus is to investigate the space of possible \textit{closed-loop} strategies, i.e., the possible existence of a Turret strategy that could eliminate this dilemma situation and, conversely, a possible Attacker strategy that could force the Turret to face the dilemma, take a guess, and possibly make a ``wrong'' decision.


\subsection{Information Limiting Strategy}\label{sec:information_limiting}
As a form of deception, we consider how the Attackers may conceal their true capability to induce suboptimal action from the Turret, which leaves a room for the Attacker team to improve its performance compared to the full-information version of the game (i.e.~Sec.~\ref{sec:complete_info} and~\cite{vonmoll2022turret-runner-penetrator}).

In a more general setting, suppose there are $m$ different possibilities for the Attacker's capabilities. 
These capabilities can be described by the \emph{action set} or admissible controls: $\mathcal{A}_k$ for $k=1,...,m$.
We use $\mathbf{A}=\{\mathcal{A}_k\}_{k=1}^m$ to denote the {set of action sets}. 
Furthermore, the intersection of all the action sets is denoted by 
\begin{equation}
    \mathcal{A}^\star = \bigcap_{k\in\{1,...,m\}} \mathcal{A}_k.
\end{equation}

Let a policy $\gamma(\x;\mathcal{A}_k,\mathbf{A})$ be a mapping from the state $\x$ to the action set $\mathcal{A}_k$. The second argument explicitly denotes the action set of the agent actually employing this policy, and the third argument signifies that the policy takes into account other possibilities for the action set.
The following two conditions guarantee that the strategy $\gamma^\star$ is information limiting:

C1) $\gamma^\star(\x;\mathcal{A}_k,\mathbf{A}) \in \mathcal{A}^\star$ for all $\x$ and $k$; and 

C2) $\gamma^\star(\x;\mathcal{A}_k,\mathbf{A}) = \gamma^\star(\x;\mathcal{A}_l,\mathbf{A})$ for all $\x$ and $k,l$.\\
The first condition ensures that the action itself does not eliminate the possibility of all $m$ cases.
The second condition ensures that the observation of the action does not help the opponent distinguish between the $m$ cases.\footnote{Note that (C2) is a stronger condition, implying (C1).}
Although these conditions may not be necessary, they are clearly sufficient for a strategy to be information limiting.

One could make further generalizations. For example, one could consider heterogeneous agents, where each one has its own set of possible action sets. 
Possible intersections and opportunities for/constraints on information limiting strategies would be more complicated than those in this paper.
Additionally, different agents could have different \textit{true} action sets.
In this paper, in accordance with Assumption~\ref{assm:two_options}, both Attackers have the same set of possible action sets (i.e., the same $\mathbf{A}$) and the same \textit{true} action set (i.e., either both Attackers are fast or both attackers are slow).

Moreover, for the problem in this paper, the control action of the Attacker is its instantaneous velocity vector denoted by $\vel{A_k}(t)=\gamma_k(\x(t);\mathcal{A}_k,\mathbf{A})$.
We define a class of strategies to be \emph{max speed information limiting} as follows.
\begin{definition}
Let the action sets $\mathcal{A}_\text{\rm slow}$ and $\mathcal{A}_\text{\rm fast}$ denote the disks with radius \emph{$\nuslow$} and \emph{$\nufast$} centered at the origin of the velocity space, and let $\mathbf{A}=\{\mathcal{A}_\text{\rm slow}, \mathcal{A}_\text{\rm fast}\}$.
The Attacker strategy $\gamma_k(\cdot)$ is information limiting (i.e., it does not reveal its true speed) 
if the following two conditions are satisfied.

C1) It only requires the slow speed:
\begin{equation}\label{eq:speed_condition_c1}
    \|\vel{A_k}(t)\|\leq \nu_{\rm slow}.
\end{equation}

C2) It is independent of $\nu$:
\begin{equation}\label{eq:policyConditionC2}
    \vel{A_k}(t)=\gamma_k(\x(t);\mathcal{A}_\text{\rm slow},\mathbf{A})=\gamma_k(\x(t);\mathcal{A}_\text{\rm fast},\mathbf{A}).
\end{equation}
\end{definition}
Since $\mathcal{A}_\text{fast}$ contains $\mathcal{A}_\text{slow}$, we have $\mathcal{A}^\star = \mathcal{A}_\text{slow}$.
%
If each Attacker moves at or below the slow speed $\nuslow$, the Turret must play the game with two possibilities in mind: (i) $\nuslow$ is the actual maximum speed, or (ii) the Attackers can actually move faster, but are moving slowly to conceal their true capability.
Furthermore, if the action / velocity is independent of $\nu$, then the Turret cannot infer the Attacker's maximum speed by observing the history of the positions nor the velocity vectors.



\section{Complete-information Game}\label{sec:complete_info}
\noindent We build our work on the \emph{Turret Attacker differential game} played with complete information on the speed ratio. This section reviews the equilibrium strategies and the Value of the game for one Attacker~\cite{vonmoll2022circular} and two Attacker~\cite{vonmoll2022turret-runner-penetrator} scenarios.
Additionally, we construct winning regions in the typical manner by computing the barrier surfaces~\cite{isaacs1965differential}.
Finally, we note some properties about how these regions can change over time while the game is played.

\begin{definition}[Complete Information Barrier Surface]
For a given (complete information) Game of Kind, the barrier surfaces divide the state space into Turret-winning and Attacker-winning regions. 
If the state is in the Turret-winning (resp. Attacker-winning) region, then the Turret (resp. Attacker) has a strategy to win the game.
\end{definition}

A graphical depiction of the barrier surface for a given \emph{Turret} location is shown in Figure 6 of~\cite{vonmoll2022turret-runner-penetrator}. However, for this paper, it is more intuitively useful to consider a given set of \emph{Attacker} positions and visualize the winning regions in $S^1$ (a circle) where the Turret state lives.

\subsection{One Attacker}
\noindent
When the game is played between a Turret and one Attacker, the Attacker's radial position $r$ and the relative angle $\theta_{A/T}= \theta_A-\theta_T$ completely characterize the system state.
The \emph{barrier surface} that divides the state space into Turret-winning and Attacker-winning regions is given by the zero-level set of the following function~\cite{shishika2018local-game}:
\begin{equation}
  \vone(r,\theta_{A/T};\nu)= |\theta_{A/T}|+F(1;\nu)-F(r;\nu),
\end{equation}
where
\begin{equation}
    F(r;\nu) = \sqrt{\left(\frac{r}{\nu}\right)^2-1}-\cos^{-1}\left(\frac{\nu }{r}\right).
\end{equation}
$\vone(r,\theta_{A/T};\nu)$ is in fact the Value of the game when the payoff is selected to be the angular distance between the agents at the time the Attacker breaches the perimeter (the Attacker is the maximizer and the Turret is the minimizer)~\cite{vonmoll2022circular}.

The optimal strategies are given as follows:
\begin{equation}
    \omega_T^* = \sgn(\theta_{A/T})
\end{equation}
\begin{equation}\label{eq:optimal_1v1_Attacker}
    \speed_{A}^* = \nu,\quad\text{and}\quad \phi_A^* = \sgn(\theta_{A/T})\sin^{-1}\left(\frac{\nu}{r}\right).
\end{equation}
It was shown in~\cite{shishika2018local-game, vonmoll2022circular} that the Attacker's heading under this equilibrium strategy points to the tangent point on the circle with radius $\nu$.

\begin{figure}[t]
    \centering
    \includegraphics[width=0.3\textwidth]{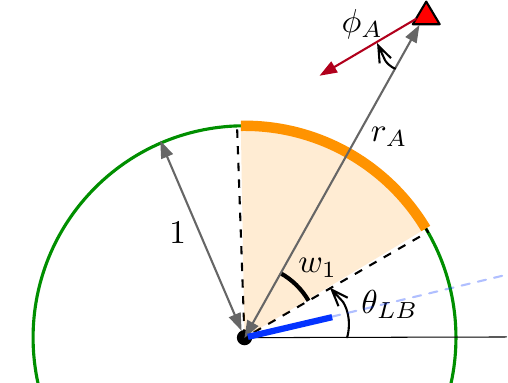}
    \caption{One vs. one Turret-winning region; in the Turret position shown, it cannot guarantee capture of the Attacker.}
    \label{fig:1v1TurretWinningRegion}
\end{figure}

Given an Attacker position $(r_{A},\theta_{A})$, the Turret-winning region (depicted in Figure~\ref{fig:1v1TurretWinningRegion}) is given by
\begin{equation}
    \region{A}(\nu) = \{\theta_T\;|\;|\theta_{A/T}|\leq 
    \widthone(r_{A};\nu)\},
\end{equation}
where $\widthone$ is a solution to $\vone(r,\theta_{A/T};\nu)=0$ w.r.t. $\theta_{A/T}$, i.e.,
\begin{equation}
    \widthone(r_{A};\nu) \triangleq F(r_{A};\nu) - F(1;\nu).
    \label{eq:widthone}
\end{equation}

We next consider how the Turret-winning region moves over time. Since we are ultimately interested in cases where the Attackers may take trajectories which are not optimal for the complete information game (because they would like to be deceptive instead), the following analysis describes how the Turret-winning region can change in response to Attacker movement at any given speed.
Suppose without loss of generality that the Turret is on the CW side of the Attacker: $\theta_{A/T}\in[0,\pi]$.
Then the relevant boundary of the Turret-winning region is given by
\begin{equation}
    \theta_{LB}(r_{A},\theta_{A};\nu) \triangleq \theta_{A} - \widthone(r_{A};\nu).
\end{equation}

%
%
\begin{lemma}\label{lem:1v1boundary}
For any given Attacker speed $v_A(t)$, the rate $\dot{\theta}_{LB}(r_A,\theta_A;\nu)$ is maximized if the Attacker's heading is tangent to the circle with radius $\nu$, as given by \eqref{eq:optimal_1v1_Attacker}.
Moreover, the maximum rate is
\begin{equation}
    \max_{\phi_A} \dot{\theta}_{LB} = \frac{v_A(t)}{\nu}.
\end{equation}
\end{lemma}
\begin{proof}
By using the following expressions
\begin{equation}
    \dot{\theta}_A=\frac{v_A\sin\phi_A}{r_A},\;\;\dot{r}_A=-v_A\cos\phi_A,\;\;\frac{\partial F}{\partial r_A}=\sqrt{\frac{1}{\nu^2}-\frac{1}{r_A^2}},
\end{equation}
the time derivative of $\theta_{LB}$ is given as
\begin{align}
    \dot{\theta}_{LB} &= \dot{\theta}_A - \dot{\widthone} 
        =\dot{\theta}_A -\dot{F}(r_A;\nu) \notag \\
        & = \frac{v_A\sin\phi_A}{r_A} + v_A\cos\phi_A \sqrt{\frac{1}{\nu^2}-\frac{1}{r_A^2}}.
\end{align}
The right-hand side is maximized when $[\cos\phi_A,\sin\phi_A]$ and $\left[\sqrt{\frac{1}{\nu^2}-\frac{1}{r_A^2}},\frac{1}{r_A}\right]$ are parallel: i.e.,
\begin{equation}
    \cos\phi_A = \sin\phi_A\sqrt{\left(\frac{r_A}{\nu}\right)^2-1}.
\end{equation}
Note that this condition is independent of the actual speed $v_A$.
Substituting this into $\sin^2\phi_A+\cos^2\phi_A=1$ yields
\begin{equation}
    \sin^2\phi_A = \left(\frac{\nu}{r_A}\right)^2.
\end{equation}
The sign of $\phi_A$ is obvious from the engagement geometry, and we obtain
\begin{equation}
    \sin\phi_A = \pm\frac{\nu}{r_A},
\end{equation}
which matches \eqref{eq:optimal_1v1_Attacker} (the sign is selected so that the Attacker has a component of velocity away from the Turret).

Now, by directly substituting the expressions for the optimal $\sin\phi_A$ and $\cos\phi_A$, we get
\begin{align}
    \dot{\theta}_{LB} &= \frac{v_A\nu}{r_A^2}+\frac{v_A\nu}{r_A}\sqrt{\left(\frac{r_A}{\nu}\right)^2-1}\sqrt{\frac{1}{\nu^2}-\frac{1}{r_A^2}} =\frac{v_A}{\nu}.
\end{align}
Importantly, this maximum rate is independent of the states.
Therefore, the heading in \eqref{eq:optimal_1v1_Attacker} provides a ``global'' optimal (not just instantaneous).
\end{proof}

\subsection{Two Attackers}
\noindent When there are two Attackers, the case where one or both Attackers can individually guarantee breach against the Turret is trivial.
The situation of interest is when $\vone<0$ for both Attackers.
The Turret can capture at least one Attacker, but we want to know whether both Attackers will be captured or one Attacker can breach.

As a form of cooperation, it was shown in~\cite{vonmoll2022turret-runner-penetrator} that one Attacker must behave as a \emph{Runner}, who sacrifices itself for the benefit of the other Attacker, who plays the role of a \emph{Penetrator}.
{We assume in this paper that the initial condition is such that, under optimal strategies, the Turret must go in one direction to capture the Runner and the opposite direction to capture the Penetrator.}
It was shown in \cite{vonmoll2022turret-runner-penetrator} that the optimal strategy for the Turret is to move towards the Runner at maximum speed.
Let us assume for now that the Turret captures \att{i} first and then \att{j}.
The Value function for the 2v1 game is given by
\begin{equation}
    \label{eq:V2v1}
  \vtwo(\x;\nu)= \vone(\pos{A_j},\dang;\nu) + 2\Runnercap(\pos{A_i},\dang;\nu)
\end{equation}
where the function $\Runnercap$ gives the angle that the Turret must travel before capturing the Runner, \att{i} (see Fig.~\ref{fig:complete_information_game}).
This angle is the solution to the following transcendental equation~\cite{vonmoll2022turret-runner-penetrator}:
\begin{equation}\label{eq:Runnercap}
 \rai \sin(\Runnercap - \theta_{A_i/T}) = \nu \Runnercap.
\end{equation}
Similar to the one-Attacker case, $\vtwo$ is the Value function of the game for which the terminal payoff is the angular separation between the Penetrator and the Turret at the time that the former breaches the target.

There is a singularity present in the solution of the two-Attacker game which is not reflected in~\eqref{eq:V2v1} (c.f.~\cite{vonmoll2022turret-runner-penetrator}).
When the angle $\theta_T\left(t_{F_i}\right) + \pi$ is contained in the reachable set of the Penetrator the singularity must be accounted for and~\eqref{eq:V2v1} does not apply.
However, for the remainder of this paper, \eqref{eq:V2v1} and~\eqref{eq:Runnercap} will be used since the regions for which deception is possible roughly correspond to the region of the two-Attacker game for which the singularity is not present.
The optimal Attacker strategies in this region are
\begin{equation}
    \label{eq:optimal_2v1_Attacker}
    \begin{aligned}
        \phi_{A_i}^* &= \sgn\left(\theta_{A_i/T}\right) \left(\frac{\pi}{2} - \tilde{\theta} + \theta_{A_i/T}\right),\\
        \phi_{A_j}^* &= \sgn\left(\theta_{A_j/T}\right) \sin^{-1}\left(\frac{\nu}{r_{A_j}}\right).
    \end{aligned}
\end{equation}
This Attacker strategy corresponds to the Runner's trajectory being perpendicular to the Turret's line-of-sight at the time of its capture and the Penetrator using its 1v1 strategy.
\begin{figure}
    \centering
    \includegraphics[width=0.3\textwidth]{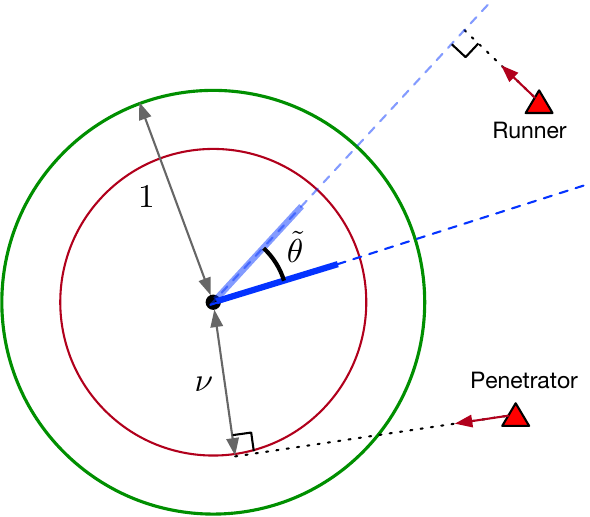} 
    \caption{The illustration of the equilibrium strategies for the Runner and Penetrator.}
    \label{fig:complete_information_game}
\end{figure}



We use $\region{A_i,A_j}(\nu)$ to denote the Turret-winning region corresponding to the capture order $A_i\rightarrow A_j$ when it is assumed that the two Attackers move at the speed $\nu$ in the remainder of the game. Note that the computation of this region can only be done numerically since it relies on the solution of~\eqref{eq:Runnercap}.

\begin{lemma}\label{lem:2v1boundary}
If the Attackers play the complete information 2v1 game optimally by deploying \eqref{eq:optimal_2v1_Attacker}, then the boundary of $\region{A_i,A_j}(\nu)$ moves at unit speed.
\end{lemma}
\begin{proof}
Consider a hypothetical Turret located exactly on the boundary of $\region{A_i,A_j}(\nu)$. By construction, this point is the boundary of the 2v1 game of kind. Moreover, assuming both sides play the optimal 2v1 game from~\cite{vonmoll2022turret-runner-penetrator}, the game state remains on the boundary of the game of kind until termination, and thus the boundary of $\region{A_i,A_j}(\nu)$ is always coincident with the hypothetical Turret's location. Finally, from~\cite{vonmoll2022turret-runner-penetrator}, the hypothetical Turret moves at unit (angular) speed.
\end{proof}

\section{No Dilemma Cases}
\label{sec:no_dilemma_cases}

\noindent 
We return now to the problem defined in Section~\ref{sec:problemFormulation}.
This section identifies the scenario of our interest by cutting out obvious cases where the Turret has a clear choice for its strategy.
As described in the motivating example, for $\nuslow$, the important winning regions are the 2v1 regions (the barrier surface describing the 2v1 winning regions determines whether $J=0$ or $J=1$), while for $\nufast$, the important winning regions are the 1v1 regions (the barrier surface describing the 1v1 winning regions determines whether $J=1$ or $J=2$). Thus, there are four important winning regions to consider: $\region{A_1}(\nufast)$, $\region{A_2}(\nufast)$, $\region{A_1,A_2}(\nuslow)$, and $\region{A_2,A_1}(\nuslow)$.
The following definitions are useful.

\begin{definition}[1v1 Intersection and Union]
We define the 1v1 intersection and union region as follows:
\begin{equation}
    \intersectone \triangleq \region{A_1} \cap \region{A_2},\;\;\;
    \unionone \triangleq \region{A_1} \cup \region{A_2}.
\end{equation}
\end{definition}

\begin{definition}[2v1 Intersection and Union]
We define the 2v1 intersection and union region as follows:
\begin{equation}
    \intersecttwo \triangleq  \region{A_1,A_2} \cap \region{A_2,A_1},\;\;
    \uniontwo \triangleq  \region{A_1,A_2} \cup \region{A_2,A_1}.
\end{equation}
\end{definition}

\subsection{Trivial Cases}
\begin{figure}
    \centering
    \includegraphics[width = 0.49\textwidth]{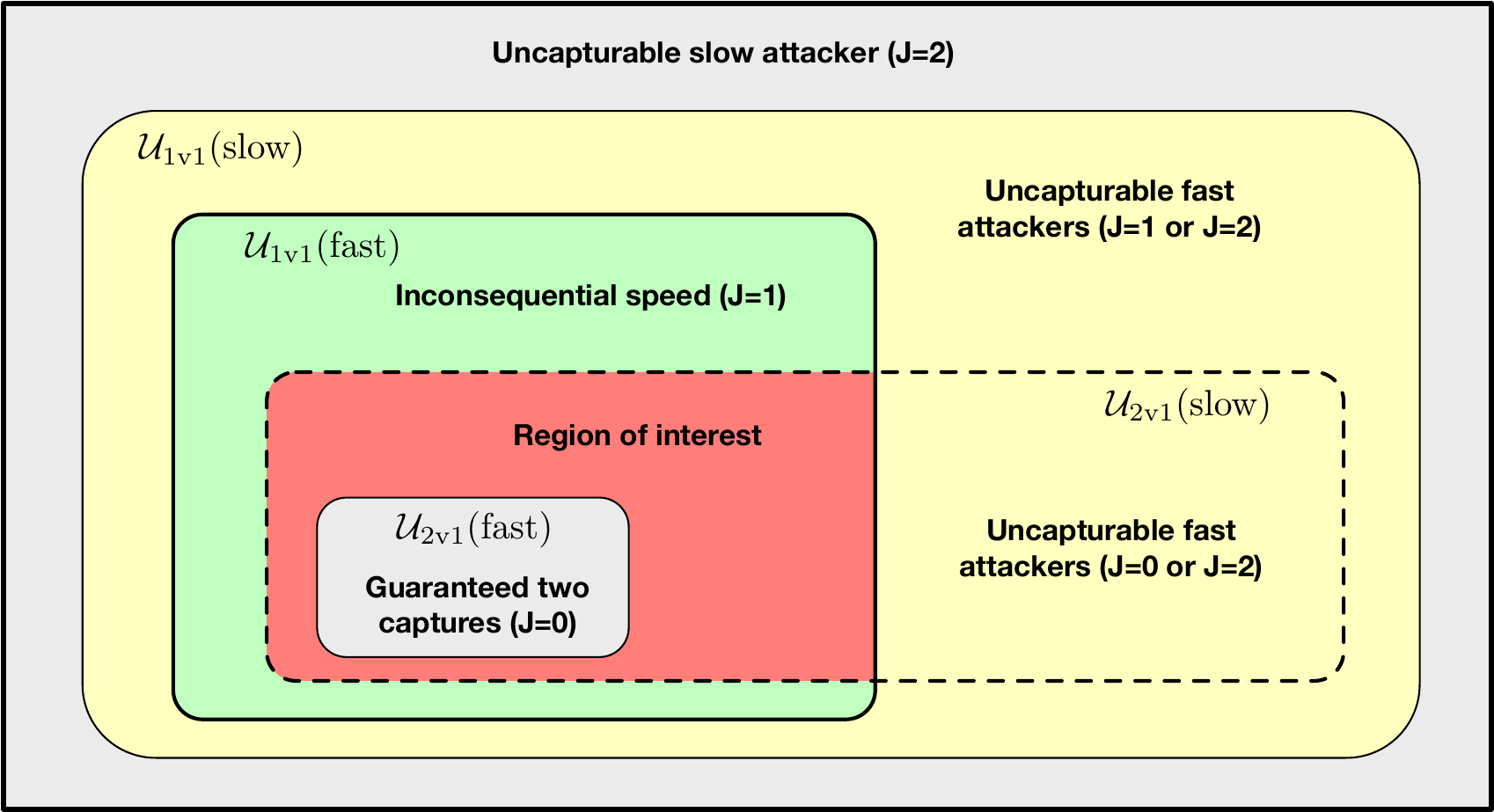}
    \caption{Venn diagram of trivial cases.}
    \label{fig:venn_diagram}
\end{figure}
\paragraph*{Guaranteed Two Captures}
If the Turret can capture two Attackers even with $\nufast$, i.e.,
\begin{equation}
    \theta_T \in \uniontwo(\nufast),
\end{equation}
then the Turret will employ the equilibrium 2v1 strategy corresponding to $\nufast$, and the outcome will be $J=0$.

\paragraph*{Uncapturable Slow Attackers}
If the states are such that
\begin{equation}
    \theta_T \notin \unionone(\nuslow),
\end{equation}
the Turret cannot capture any Attacker.
In this case the Turret's strategy is inconsequential in terms of the outcome: i.e., as long as the Attackers are rational, the outcome $J=2$ is guaranteed.

\paragraph*{Uncapturable Fast Attackers}
If the states are such that
\begin{equation}
    \theta_T \notin \unionone(\nufast),
\end{equation}
then the Turret cannot capture any Attacker if the Attackers are fast and rational.
Therefore, the Turret risks nothing by disregarding the possibility of $\nufast$ and will play the game assuming $\nuslow$: i.e., play the 2v1 slow if $\theta_T \in \uniontwo(\nuslow)$, otherwise, play the 1v1 slow game.
If the Attackers are actually slow, these selections are optimal. If the Attackers are actually fast, the Turret's strategy is inconsequential.

\paragraph*{Inconsequential Speed}
If the speeds $\nuslow$ and $\nufast$ both lead to one capture, i.e.,
\begin{equation}
    \theta_T \in \unionone(\nufast)\setminus  \uniontwo(\nuslow)
\end{equation}
then the Turret can only capture one Attacker regardless of the speed ratio.
In this case, the Turret will play the 1v1 fast game to ensure $J \leq 1$. If the Attackers also play their corresponding equilibrium strategy, the outcome will be $J=1$.

\subsection{Matching Directions}
\label{sec:matching_directions}

\noindent
Based on the observations above, we can restrict our attention to the states where 
\begin{equation}
    \theta_T \in \unionone(\nufast)\cap \uniontwo(\nuslow) \setminus \uniontwo(\nufast),
    \label{eq:nec_for_dilemma}
\end{equation}
i.e., one capture is possible if $\nu=\nufast$, and two captures are possible if $\nu=\nuslow$.  
All other cases are covered in the previous subsection.
Figure~\ref{fig:venn_diagram} illustrates the relationship between the region of interest (red) and other trivial cases.

A key feature of the motivating example was that the optimal strategies corresponding to the $\nuslow$ and $\nufast$ games were at odds (CW for $\nuslow$ and CCW for $\nufast$).
The following lemma describes a case where this tension does not exist.

\begin{lemma}
\label{lem:matching_directions}
Consider the case with \eqref{eq:nec_for_dilemma}.
If the states also satisfy
\begin{equation}
    \theta_T \in \region{A_i,A_j}(\nuslow) \cap \region{A_i}(\nufast),
    \label{eq:matching_directions}
\end{equation}
then the Turret has a deterministic strategy which ensures two captures if $\nu=\nuslow$ and one capture if $\nu=\nufast$.
\end{lemma}
\begin{proof}
In this configuration the slow game and the fast game share the same equilibrium Turret strategy: move towards $A_i$ at max speed.
The Turret can continue to play optimally without making a decision of whether it is playing the slow game or the fast game.
\end{proof}

\subsection{Attacker Strategy}
\noindent
For each of these no-dilemma cases, the Turret can guarantee the associated outcomes based on the solution of the associated differential game.
These no-dilemma differential game solutions satisfy a saddle-point equilibrium property which means the Turret's guarantees apply to the ``worst-case" Attacker inputs.
Therefore there is no incentive for the Attackers to attempt to conceal their true capability.

\section{Guaranteed Dilemma Case}
\label{sec:guaranteed_dilemma}
\noindent 
Section~\ref{sec:matching_directions} considered the case when \eqref{eq:nec_for_dilemma} and \eqref{eq:matching_directions} hold.
Here we visit the case when the Turret's strategies corresponding to 2v1-slow game and 1v1-fast game do not match, i.e., when \eqref{eq:matching_directions} does \emph{not} hold.

\begin{definition}[The \AttackerStrategy]
We define the Slow-Speed Fast-Heading (\AttackerStrategy) as follows:
\begin{itemize}
    \item Both Attackers move at $\nuslow$.
    \item The direction of motion corresponds to the 1v1 fast game, i.e., move straight towards the tangent point of $\nufast$ circle.
\end{itemize}
\end{definition}
Note that this strategy can be implemented regardless of the Attacker's true maximum speed, and is therefore information-limiting according to \eqref{eq:speed_condition_c1}. Through the remainder of the paper, every time the SS-FH strategy is used, it is done regardless of the true maximum speed and is thus also information-limiting according to \eqref{eq:policyConditionC2}.

\begin{figure}[t]
    \centering
    \includegraphics[width=0.35\textwidth]{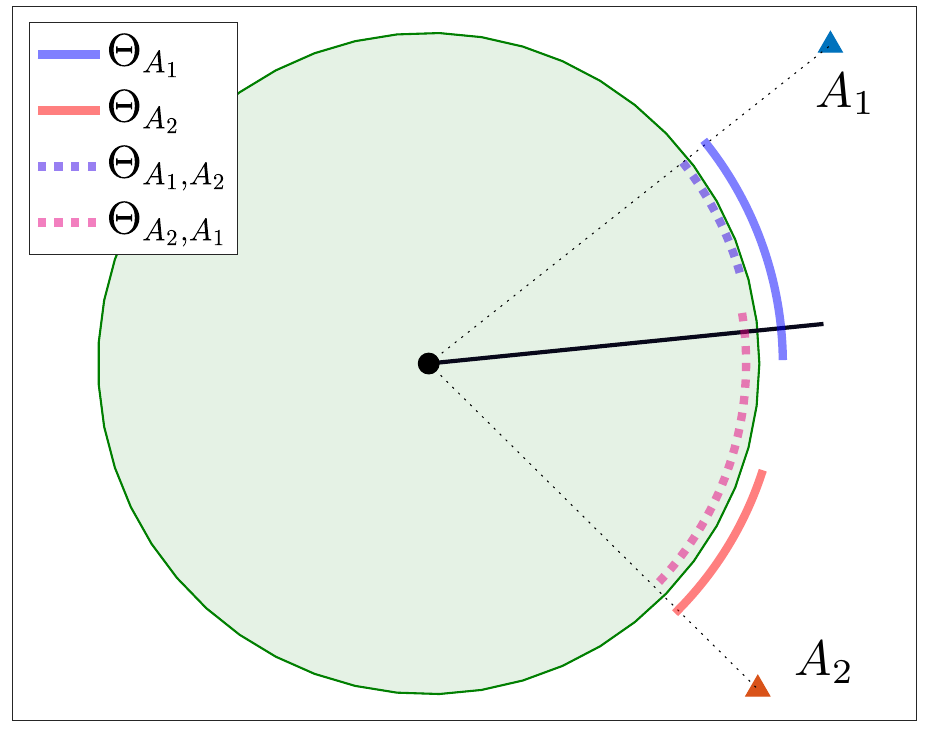}
    \caption{Configuration leading to dilemma. $\nuslow=0.2$, $\nufast=0.7$.}
    \label{fig:guaranteed_dilemma}
\end{figure}

\begin{theorem}
\label{thm:guaranteed_dilemma}
Consider an initial state that satisfies~\eqref{eq:nec_for_dilemma}.
If the intersection between the two regions, $\unionone(\nufast)$ and $\uniontwo(\nuslow)$, correspond to conflicting directions: i.e., 
\begin{equation}
    \theta_T \in \region{A_i,A_j}(\nuslow)\cap\region{A_j}(\nufast),
    \label{eq:mismatch_overlap}
\end{equation}
and also if
\begin{equation}
    \intersecttwo(\nuslow)=\intersectone(\nufast)=\emptyset,
    \label{eq:no_overlap}
\end{equation}
then the Attackers have an information-limiting strategy to generate a dilemma.
\end{theorem}
\begin{proof}
Without loss of generality let $A_i=A_2$ be in the CW direction and $A_j=A_1$ be in the CCW direction. The posited configuration is depicted in Figure~\ref{fig:guaranteed_dilemma}.
Also suppose that the Attackers employ the \AttackerStrategy.

First, observe that the condition \eqref{eq:no_overlap} remains true if the Attackers employ the \AttackerStrategy.
Noting that $\region{A_1}(\nufast)$ extends from $\theta_{A_1}$, it is obvious that $\theta_{A_2}\notin \region{A_1}(\nufast)$ so long as $\intersectone(\nufast)=\emptyset$.
Hence, the Turret cannot capture $A_2$ while staying in $\region{A_1}(\nufast)$ if \eqref{eq:no_overlap} is true. 
Similarly, noting that $\region{A_2,A_1}(\nuslow)$ extends from $\theta_{A_2}$,
it can be seen that $\theta_{A_1} \notin \region{A_2,A_1}(\nuslow)$ so long as $\intersecttwo=\emptyset$.\footnote{If $\region{A_2,A_1}(\nuslow)$ contained $\theta_{A1}$, then a Turret could start CCW of \att{1} and just turn CW to capture both. 
This means that $\region{A_1,A2}(\nuslow)$ exists around $\theta_{A1}$ and intersects $\region{A_2,A_1}(\nuslow)$, which contradicts the condition that $\intersecttwo(\nuslow)=\emptyset$.}
Hence, the Turret cannot capture $A_1$ while staying in $\region{A_2,A_1}(\nuslow)$, if \eqref{eq:no_overlap} is true.

{The above two observations show that, if \eqref{eq:no_overlap} is true, the Turret cannot capture either of the Attackers while maintaining the condition \eqref{eq:mismatch_overlap}, i.e., staying in the ``mismatched'' overlap region.
Next, we will see that this overlap region shrinks and disappears in finite time.}


{If the Attackers use the \AttackerStrategy, the boundary of $\region{A_2,A_1}$ moves CW.}
Also, from Lemma~\ref{lem:1v1boundary}, the boundaries of $\region{A_1}$ and $\region{A_2}$ move at the velocity $+\nuslow/\nufast$ and $-\nuslow/\nufast$ respectively.
From these observations, we can see that the intersection in \eqref{eq:mismatch_overlap} monotonically decreases and disappears eventually.
This implies that, {to achieve any capture}, the Turret must leave either $\region{A_1}(\nufast)$ or $\region{A_2,A_1}(\nuslow)$, which is the dilemma that the Turret faces.

As one alternative, the Turret may choose to leave $\region{A_2,A_1}(\nuslow)$ and stay in $\region{A_1}(\nufast)$ to ensure one capture.
This choice is optimal if the speed is indeed $\nu=\nufast$.
However, if $\nu=\nuslow$, the Turret's choice is suboptimal since the Attackers can switch their strategy to optimal 2v1 slow game and avoid two captures, which might have been possible if the Turret stayed in $\region{A_2,A_1}(\nuslow)$.

As the other option, the Turret may choose to leave $\region{A_1}(\nufast)$ and stay in $\region{A_2,A_1}(\nuslow)$ to achieve two captures.
This choice is optimal if the speed is indeed $\nu=\nuslow$.
However, if $\nu=\nufast$, the Turret's choice is suboptimal since now that $\theta_T\notin\unionone(\nufast)$, the Attackers can speed up and play the 1v1 fast game to ensure that both Attackers score.

We have shown that if Attackers employ the \AttackerStrategy, the Turret must make a \emph{guess} between $\nuslow$ and $\nufast$.
In either case, there is a possibility of guessing wrong and missing one capture that could have been achieved if the true speed was known.
\end{proof}




\begin{remark}[No Risk for Attackers]
If the Attackers are actually slow, and~\eqref{eq:mismatch_overlap}, \eqref{eq:no_overlap} are satisfied, then the Attackers cannot guarantee any breaches; however if, by employing the \AttackerStrategy, the Turret decides to stay in $\region{A_1}(\nufast)$, then $A_2$ may be able to breach.
Similarly, if the Attackers are actually fast, they can only guarantee one breach; but if the Turret hedges against $\nuslow$, then \emph{both} Attackers may be able to breach by speeding up once $\theta_T \notin \region{A_1}(\nufast)$.
The Attackers risk nothing by employing the \AttackerStrategy.
\end{remark}

\section{Intermediate Configuration}\label{sec:intermediate}
\noindent
In Theorem~\ref{thm:guaranteed_dilemma}, the condition \eqref{eq:no_overlap} ensured that the Turret cannot enter a matching condition \eqref{eq:matching_directions} without risking a suboptimal outcome by leaving either $\region{A_i,A_j}(\nuslow)$ or $\region{A_j}(\nufast)$.
Whether such dilemma arises or not becomes more subtle when the overlap regions are nonempty, i.e., when \eqref{eq:no_overlap} does not hold.
The problem is that while \eqref{eq:matching_directions} may not hold at the initial time, there may exist a Turret strategy to accomplish \eqref{eq:matching_directions} at some later time while maintaining \eqref{eq:mismatch_overlap} constantly during this period. That is, it may be possible for the Turret to enter a ``matched" overlap region in a way that resolves any dilemma before the critical moment when it would have had to guess.\footnote{Animations available at YouTube \url{https://youtu.be/n0C0UZqJmMU}\label{fn:video}
}

There are three possible ways for \eqref{eq:no_overlap} to not hold:
\begin{enumerate}
    \item $\intersecttwo(\nuslow)=\emptyset,\intersectone(\nufast)\neq\emptyset$
    \item $\intersecttwo(\nuslow)\neq\emptyset,\intersectone(\nufast)=\emptyset$
    \item $\intersecttwo(\nuslow)\neq\emptyset,\intersectone(\nufast)\neq\emptyset$
\end{enumerate}
We will first focus on case 1 and show Turret sufficiency. That is, we will show a condition under which there exists a Turret strategy to guarantee its ability to enter $\intersectone(\nufast)$ and resolve the dilemma before making a guess. Then, we will turn to cases 2-3 and prove necessity. That is, we will show that even if both overlap regions exist, if the condition does not hold for Turret sufficiency, there exists an Attacker strategy to force the dilemma. Case 1 will focus on the usefulness of the SS-FH strategy, Case 2 will motivate the usefulness of the slow 2v1 game strategy corresponding to~\eqref{eq:optimal_2v1_Attacker}, and Case 3 will show how a combination of the two strategies can be used in general.

\subsection{Turret Avoiding Dilemma}
\label{sec:avoiding_dilemma}
\noindent
Begin with Case 1, when $\intersectone(\nufast)\neq\emptyset$.

\begin{definition}[1v1 Intersection Reachability Set]\label{def:overlap_reachabilityone}
Let $\cone$ be the center of $\intersectone(\nufast)$, and $2\wone$ be its width.
We define the \emph{1v1 overlap reachability set} as
\begin{equation}
    \reachone = \{\theta:\; \alpha|\theta-\cone|\leq\wone\},
\end{equation}
where $\alpha = \nuslow/\nufast$.
\end{definition}
\begin{lemma}
\label{lem:1v1overlap_reachability}

The Turret has a strategy to enter $\intersectone(\nufast)$ against all information-limiting Attacker strategies iff $\theta_T\in\reachone$.
\end{lemma}
\begin{proof}
For the strategy to be information-limiting, it is necessary for the Attackers to move at a speed less than or equal to $\nuslow$ {according to \eqref{eq:speed_condition_c1}}.
From Lemma~\ref{lem:1v1boundary}, the maximum speed at which the boundaries of $\intersectone(\nufast)$ can move is $\alpha$.
The Turret can enter $\intersectone(\nufast)$ iff it can reach $\cone$ before $\intersectone(\nufast)$ disappears in $\wone/\alpha$ units of time.
\end{proof}

\begin{lemma}\label{lem:vel_boundary_R1v1}
    If the Attackers are restricted to use information-limiting strategies,
    then the boundary of $\reachone$ moves at most at unit speed. 
    Moreover, the \AttackerStrategy{} (resp.~velocity antiparallel to SS-FH) is the only strategy that achieves this maximum speed in the direction that shrinks (resp.~expands)~$\reachone$.
\end{lemma}
\begin{proof}
Let $\theta_1$ and $\theta_2$ denote the boundaries of $\region{A_1}(\nufast)$ and $\region{A_2}(\nufast)$ that comprise the boundaries of $\intersectone(\nufast)$ as shown in Fig.~\ref{fig:avoidDilemma}. 
If the Attackers use information limiting strategies, then their speeds are limited by $\nuslow$, and consequently, the result of Lemma~\ref{lem:1v1boundary} gives us the following bounds: $|\dot{\theta}_1|\leq \alpha$ and $|\dot{\theta}_2|\leq \alpha$.

The center and the width of $\intersectone$ is given as $\cone=0.5(\theta_1+\theta_2)$ and $\wone=0.5(\theta_2-\theta_1)$.
Now, the boundary of $\reachone$ that is relevant to our problem is
\begin{equation}
    \theta_{UB} = \cone + \frac{1}{\alpha}\wone,
\end{equation}
whose time derivative can be described in terms of $\dot{\theta}_1$ and $\dot{\theta}_2$ as follows:
\begin{align}\label{eq:thetaUB_dot}
    \dot{\theta}_{UB} &= \dcone + \frac{1}{\alpha}\dwone\\
    &= 0.5\left( 1-\frac{1}{\alpha} \right)\dot{\theta}_1+0.5\left(1+\frac{1}{\alpha}\right)\dot{\theta}_2.
\end{align}
Noting that $(1-1/\alpha)<0$, the maximum clockwise velocity of $\theta_{UB}$ is achieved when $\dot{\theta}_1 = -\dot{\theta}_2 = \alpha$, which gives $\dot{\theta}_{UB}=-1$ ($\reachone$ shrinks). Similarly, the maximum counter-clockwise velocity of $\theta_{UB}$ is achieved when $-\dot{\theta}_1 = \dot{\theta}_2 = \alpha$, which gives $\dot{\theta}_{UB}=+1$ ($\reachone$ expands).

Finally, one can verify from Lemma~\ref{lem:1v1boundary} that the \AttackerStrategy{} is the only one that achieves $\dot{\theta}_1 = -\dot{\theta}_2 = \alpha$, and the velocities anti-parallel to the \AttackerStrategy{} is the only one that achieves $-\dot{\theta}_1 = \dot{\theta}_2 = \alpha$.
\end{proof}

\begin{figure}[t]
    \centering
    \includegraphics[width=0.4\textwidth]{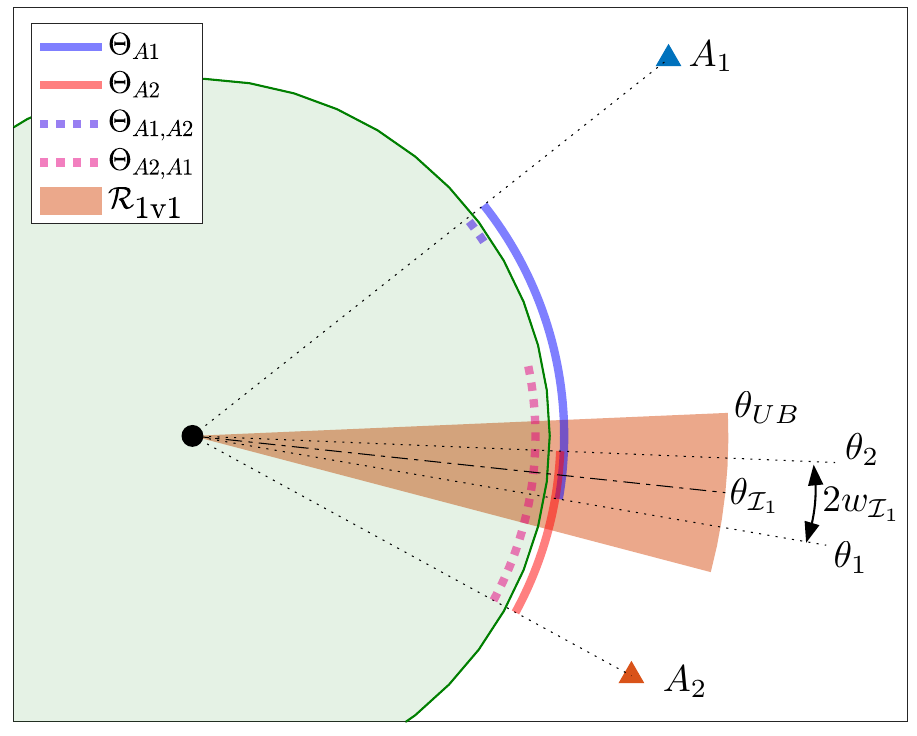}
    \caption{Illustration of the 1v1 overlap reachability set, $\reachone$, for parameters $\nuslow=0.3$ and $\nufast=0.7$.}
    \label{fig:avoidDilemma}
\end{figure}

\begin{theorem}\label{thm:avoid_dilemma}
Suppose the initial states satisfy \eqref{eq:mismatch_overlap}. 
If the Turret is also in the 1v1 intersection reachability set, i.e., 
\begin{equation}
    \theta_T \in \reachone,
    \label{eq:avoid_dilemma_sufficient}
\end{equation}
then it can avoid the dilemma situation and ensure optimal number of captures (i.e., two if $\nutrue = \nuslow$ and one if $\nutrue=\nufast$).
Consequently, the Attackers' knowledge of the true maximum-speed information has no value.
\end{theorem}
\begin{proof}
First, consider the case where the Attackers are fast and they reveal their true speed while condition \eqref{eq:mismatch_overlap} is still satisfied.
Since $\region{A_j}(\nufast)$ still holds, the Turret is able to ensure one capture (optimal number) by moving towards $A_j$ after observing that the Attackers are fast.

Now, consider the case where the Attackers use $\nuslow$ regardless of their true capability as long as \eqref{eq:mismatch_overlap} holds.
By construction, the Turret can enter $\intersectone$ before it disappears (Lemma~\ref{lem:1v1overlap_reachability}).
Along the way, the Turret moves at max speed towards $A_i$, which is also the optimal behavior for 2v1 slow game with the capture order $A_i\rightarrow A_j$.
Therefore, we know that the condition $\theta_T\in\region{A_i, A_j}$ still holds.
This result implies that the Turret is now in the matching condition in \eqref{eq:matching_directions}.
Together with Lemma~\ref{lem:matching_directions}, the proof is complete.
\end{proof}


\subsection{Attackers Forcing Dilemma} \label{sec:forcing_dilemma}
\noindent
This section shows that the condition \eqref{eq:avoid_dilemma_sufficient} is not only sufficient but also necessary for the Turret to avoid dilemma.

Begin with Case 2, where $\intersecttwo(\nuslow)\neq\emptyset$.
As a counterpart of Lemma~\ref{lem:1v1overlap_reachability}, we provide the following definition for $\reachtwo$.

\begin{definition}[2v1 Intersection Reachability Set]\label{def:overlap_reachabilitytwo}
    We define the \emph{2v1 intersection reachability set} as
    \begin{equation}
    \reachtwo=\intersecttwo(\nuslow),
    \end{equation}
\end{definition}

This definition is motivated by the following remark.

\begin{remark}
Suppose $\theta_T \notin \intersecttwo(\nuslow)$. If the Attackers play the slow 2v1 game corresponding to~\eqref{eq:optimal_2v1_Attacker}, the boundary of $\intersecttwo(\nuslow)$ moves in the direction away from the Turret at unit speed in the inertial frame.
\end{remark}

It can be shown that an Attacker strategy of playing the slow 2v1 game corresponding to~\eqref{eq:optimal_2v1_Attacker} is sufficient for the attackers to force the dilemma in Case 2. We omit the formal proof, because it will become clear that Case 2 can be subsumed by the subsequent discussion on Case 3. However, it is critical to observe that while Definition~\ref{def:overlap_reachabilitytwo} is motivated by this case and the Attacker strategy of playing the slow 2v1 game corresponding to~\eqref{eq:optimal_2v1_Attacker}, both $\reachone,\reachtwo$ are defined independently of any particular state history or strategy instantiation.

Now, consider Case 3, where $\intersecttwo(\nuslow)\neq\emptyset,\intersectone(\nufast)\neq\emptyset$.
Given that Sec.~\ref{sec:avoiding_dilemma} still ensures that the Turret can avoid dilemma if it is in $\reachone$, now we assume that the Turret is \emph{not} in $\reachone$.
We can define distances $d_1(t)\geq 0$ and $d_2(t) \geq 0$ that describe the angular separation between the Turret's position and the boundaries of $\reachone$ and $\reachtwo$ respectively (see Fig.~\ref{fig:force_dilemma_overlap}).
\begin{figure}[t]
    \centering
    \includegraphics[width=0.4\textwidth]{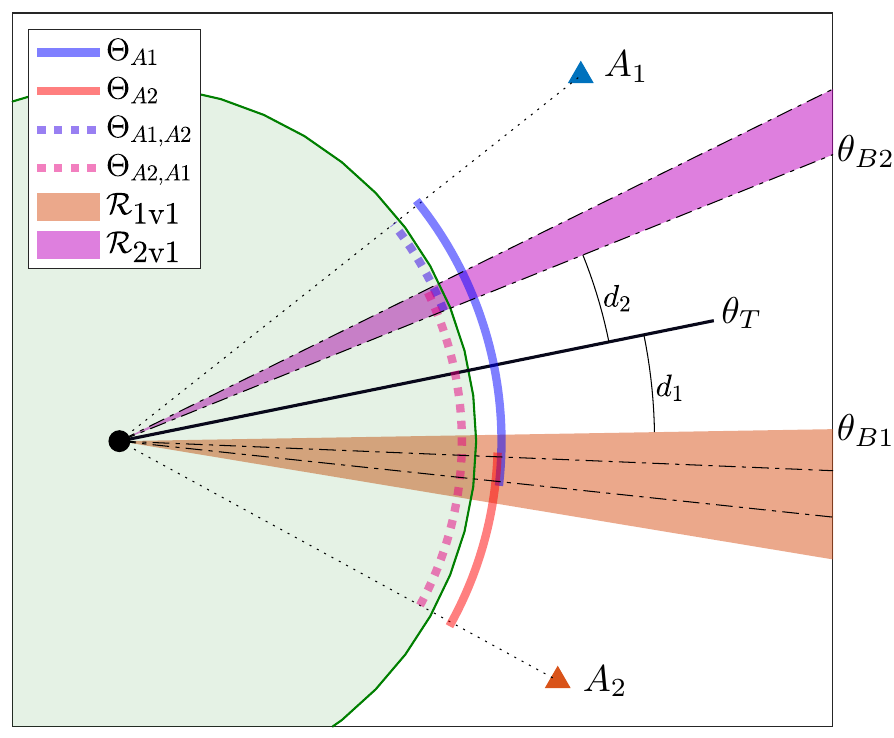}
    \caption{Illustration of Case 3 with parameters $\nuslow=0.25$, $\nufast=0.7$. 
    $\theta_{B1}$ and $\theta_{B2}$ are the angular positions of the boundaries of $\reachone$ and $\reachtwo$ respectively.}
    \label{fig:force_dilemma_overlap}
\end{figure}

Letting $\theta_{B1}$ and $\theta_{B2}$ denote the boundaries of $\reachone$ and $\reachtwo$ as in Fig.~\ref{fig:force_dilemma_overlap}, it is easy to see the following relationship:
\begin{eqnarray}
    \dot{d}_1 &=& \omega_T - \dot{\theta}_{B1} \\
    \dot{d}_2 &=& -\omega_T + \dot{\theta}_{B2}.
\end{eqnarray}
The signs are based on the configuration in Fig.~\ref{fig:force_dilemma_overlap}, which will change if $\reachone$ (resp.~$\reachtwo$) is on the CCW (resp.~CW) side of the Turret.

Before we present the main theorem, we show three Lemmas that describe how the distances $d_1(t)$ and $d_2(t)$ evolve under the two types of Attacker strategies we've considered and that arbitrary combinations of these two strategies result in the overlap regions disappearing in finite time.
\begin{lemma}\label{lem:d1d2behavior_a}
Suppose the Attackers play the 2v1 slow game.
Then we have $\dot{d}_2\geq0$ for any admissible Turret strategy. 
Moreover, we have
\begin{equation}
    \label{eq:2v1_slow_distances_growing}
    \dot{d}_1+\dot{d}_2\geq 0. \notag
\end{equation}
\end{lemma}
\begin{proof}
From Lemma~\ref{lem:2v1boundary}, the boundary of $\reachtwo$ moves away from the Turret at unit speed in the inertial frame: $\dot{\theta}_{B2}=1$. 
This leads to $\dot{d}_2=-\omega_T + 1$.
Noting that $\lvert \omega_T \rvert \leq 1$, the above equality proves the first part: $\dot{d}_2\geq0$.

For the second part, it suffices to show that $\dot{d}_1 + \dot{d}_2 \geq 0 \Leftrightarrow \dot{d}_1 - \omega_T + 1\geq 0 \Leftrightarrow \omega_T - \dot{d}_1 \leq 1 \Leftrightarrow \dot{\theta}_{B1} \leq 1$, where the quantity $\dot{\theta}_{B1} \triangleq \omega_T - \dot{d}_1$ is the velocity at which the boundary of $\reachone$ moves.
The condition $\dot{\theta}_{B1} \leq 1$ is equivalent to saying that this boundary (in the inertial frame) either moves towards the Turret at a speed less than or equal to unity or moves away from the Turret. From Lemma~\ref{lem:vel_boundary_R1v1}, this boundary only moves toward the Turret at unity speed if the Attackers move in the direction anti-parallel that of the SS-FH strategy.
\end{proof}

The following Lemma shows a similar result for the \AttackerStrategy.
\begin{lemma}\label{lem:d1d2behavior_b}
Suppose the Attackers play the \AttackerStrategy.
Then we have $\dot{d}_1\geq 0$ for any admissible Turret strategy. 
Moreover, we have
\begin{equation}
    \dot{d}_1+\dot{d}_2\geq 0. \notag
\end{equation}
\end{lemma}
\begin{proof}
From Lemma~\ref{lem:vel_boundary_R1v1}, the boundary of $\reachone$ moves away from the Turret at unit speed when the Attackers use the \AttackerStrategy,
which proves $\dot{d}_1\geq 0$.
Similar to the proof of Lemma~\ref{lem:d1d2behavior_a}, we are left to prove that the boundary of $\reachtwo$ does not move towards the Turret any faster than unit speed ($\dot{\theta}_{B2}\geq-1$).
In the following, we will show that the \AttackerStrategy\ ensures that the boundary of $\reachtwo$ in fact moves away from the Turret ($\dot{\theta}_{B2}>0$), which is more than sufficient to prove $\dot{d}_1+\dot{d}_2\geq 0$.

Let $\theta_B$ be a fictitious Turret position that coincides with the boundary of $\reachtwo$ on the CW side, satisfying
\begin{equation}\label{eq:const_boundary}
    \vtwo = \vone(r_P,\theta_P,\theta_B) + 2\Runnercap(r_R,\theta_R,\theta_B) \equiv 0,
\end{equation}
where $_P$ and $_R$ denote the Penetrator and Runner respectively.
Without loss of generality, suppose that the Penetrator is on the CW side and the Runner is on the CCW side of the Turret.
Then we have $\frac{\partial \vone}{\partial \theta_B} = 1$ and $\frac{\partial \Runnercap}{\partial \theta_B} < -1$,\footnote{This can be shown by differentiating \eqref{eq:Runnercap} as shown in the \red{Appendix}.} which results in
\begin{equation}\label{eq:turretBenefitsCloserToRunner}
    \frac{\partial \vtwo}{\partial \theta_B} < 0.
\end{equation}
The above inequality implies that a fictitious Turret benefits from starting closer to the Runner.

For the constraint \eqref{eq:const_boundary} to be satisfied through time (i.e., $\theta_B(t)=\theta_{B2}(t)$), $\theta_B$ must also satisfy $\dot{V}_\text{2v1}\equiv 0$:
\begin{eqnarray}\label{eq:const_boundary_partial}
    \dot{V}_\text{2v1}&=&\frac{\partial \vone}{\partial r_P} \dot{r}_P
    + \frac{\partial \vone}{\partial \theta_P} \dot{\theta}_P
    + \frac{\partial \vone}{\partial \theta_B} \dot{\theta}_B  \notag \\
    &+& 2\left[
    \frac{\partial \Runnercap}{\partial r_R} \dot{r}_R
    + \frac{\partial \Runnercap}{\partial \theta_R} \dot{\theta}_R
    + \frac{\partial \Runnercap}{\partial \theta_B} \dot{\theta}_B
    \right]
    \equiv 0.
\end{eqnarray}
Since moving farther out is worse for the Attacker, we have the following:
$\frac{\partial \vone}{\partial r_P}<0,\quad \frac{\partial \Runnercap}{\partial r_R} <0.$
Due to the assumed arrangement of the Penetrator and the Runner, the partials satisfy the following: 
$\frac{\partial \vone}{\partial \theta_P}<0,\quad \frac{\partial \Runnercap}{\partial \theta_R} >0.$
Finally, the \AttackerStrategy\ gives the following signs on the time derivative of the states:
$    \dot{r}_P<0,\quad \dot{\theta}_P<0,\quad \dot{r}_R<0,\quad \dot{\theta}_R>0.$

The condition \eqref{eq:const_boundary_partial} with the above inequalities provide the following result:
\begin{eqnarray}
    \text{[positive terms]} + \dot{\theta}_B\left( \frac{\partial \vone}{\partial \theta_B} + 2\cdot\frac{\partial \Runnercap}{\partial \theta_B} \right) \equiv 0 \notag \\
    \Leftrightarrow \dot{\theta}_B \cdot \frac{\partial \vtwo}{\theta_B} < 0.
\end{eqnarray}

From~\eqref{eq:turretBenefitsCloserToRunner}, $\frac{\partial \vtwo}{\theta_B} < 0$, it must be that $\dot{\theta}_B > 0$. Furthermore, since both \eqref{eq:const_boundary} and \eqref{eq:turretBenefitsCloserToRunner} are satisfied at all times, $\theta_B(t)=\theta_{B2}(t)$, and therefore the boundary of $\reachtwo$ moves away from the Turret under the \AttackerStrategy, i.e., $\dot{\theta}_{B2}>0$.
\end{proof}

One final lemma shows that we can combine the two Attacker strategies from Lemmas~\ref{lem:d1d2behavior_a} and~\ref{lem:d1d2behavior_b}, and still ensure that the overlap regions disappear in finite time, i.e., that the game necessarily reaches the guaranteed dilemma case of Section~\ref{sec:guaranteed_dilemma}.

\begin{lemma}\label{lem:arbitrarySwitchingFinite}
    Under arbitrary switching between the 2v1 slow strategy and SS-FH strategy, both regions $\reachone$ and $\reachtwo$ vanish in finite time.
\end{lemma}
\begin{proof}
    There are three relevant behaviors to consider: the SS-FH behavior, the 2v1 slow Runner behavior, and the 2v1 slow Penetrator behavior. Each corresponds to a constant, negative $\dot{r}$. If we take the maximum (least negative) of these rates, both Attackers are always approaching the perimeter at least as fast as this finite rate, regardless of which team strategy is active. Therefore, in finite time, at least one attacker will reach the perimeter. Suppose without loss of generality, it is Attacker $i$ that reaches the perimeter. At that time, the regions $\region{A_i}$ and $\region{A_i,A_j}$ must vanish, and consequently, the intersection regions cannot exit.
\end{proof}

\begin{theorem}\label{thm:forcing_dilemma}
Consider the initial configuration that satisfies \eqref{eq:nec_for_dilemma} and \eqref{eq:mismatch_overlap}.
If the Turret is not in the 1v1 overlap reachability set:
\begin{equation}
    \theta_T \notin \reachone(\nufast),
    \label{eq:guaranteed_dilemma_sufficient}
\end{equation}
then the Attackers have an information-limiting strategy to drive the system into the dilemma scenario in Theorem~\ref{thm:guaranteed_dilemma}.
\end{theorem}
\begin{proof}
First, note that \eqref{eq:nec_for_dilemma} and \eqref{eq:mismatch_overlap} imply that 
\begin{equation}
    \theta_T \notin \reachtwo(\nuslow).
\end{equation}
It suffices to show the following two: (i) the overlap regions become empty in finite time; and (ii) $d_1>0$ and $d_2>0$ for all time.
We will prove that the above holds with the following Attacker strategy:
\begin{itemize}
    \item[a1)] Play 2v1 slow game if $d_2<d_1$, and
    \item[a2)] Play the \AttackerStrategy\ if $d_1 \leq d_2$.
\end{itemize}
Note that this is not the only strategy, but there is a class of strategies that leads to the same guarantee. We can immediately note that (i) is a trivial consequence of Lemma~\ref{lem:arbitrarySwitchingFinite} and proceed to showing (ii).
\begin{figure}[ht]
    \centering
    \includegraphics[width=0.49\textwidth]{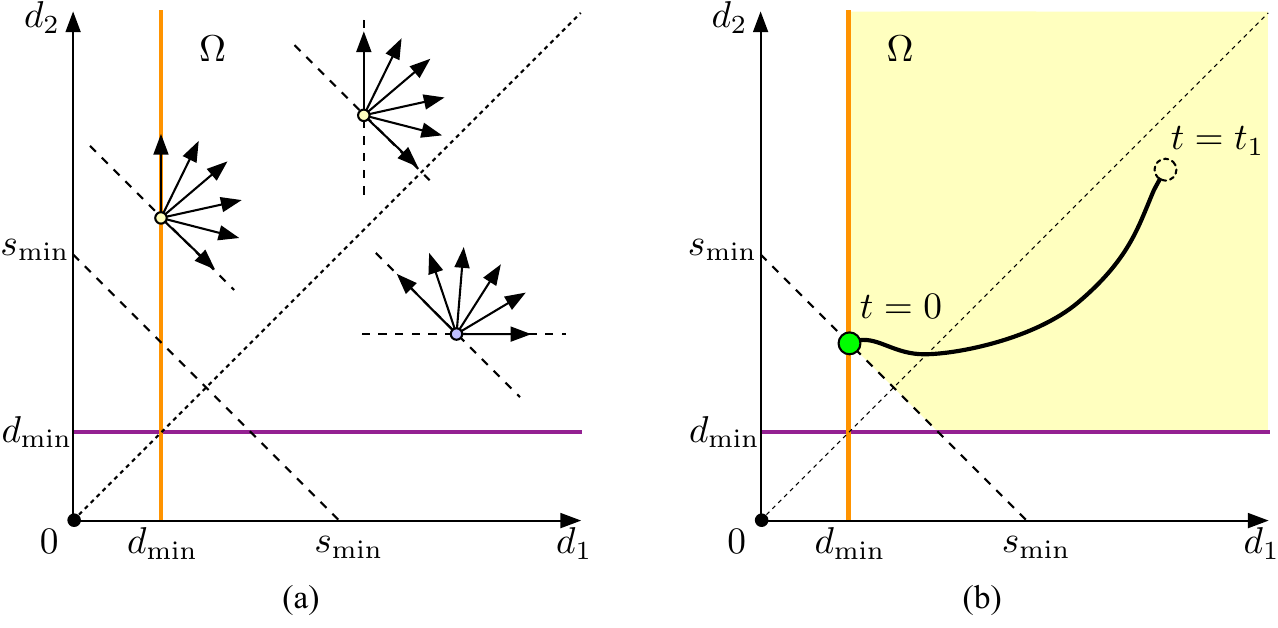}
    \caption{Illustration of the evolution of $d_1(t)$ and $d_2(t)$, and the forward invariant set $\Omega$. (a) A vectogram resulting from Lemma~\ref{lem:d1d2behavior_a} and~\ref{lem:d1d2behavior_b}. (b) The time $t_1$ indicates when either $\intersectone$ or $\intersecttwo$ disappears.}
    \label{fig:d1d2cartoon}
\end{figure}

Given an initial condition $d_1(0)$ and $d_2(0)$, consider the following set in the $d_1d_2$-space (also see Fig.~\ref{fig:d1d2cartoon}):
\begin{equation}
    \Omega = \{[d_1,d_2]\;|\; \min\{d_1,d_2\}>d_\text{min},\; d_1+d_2>s_\text{min} \},
\end{equation}
where $d_\text{min} = \min\{d_1(0),d_2(0)\}$ and $s_\text{min} = d_1(0)+d_2(0)$.
Based on Lemmas~\ref{lem:d1d2behavior_a} and \ref{lem:d1d2behavior_b}, the set $\Omega$ is forward invariant.
Therefore, neither $d_1$ or $d_2$ becomes less than $d_\text{min}$, implying that the Turret can never enter the overlap regions.
%
\end{proof}


\section{Numerical Results}
This section demonstrates the theoretical results of the paper through simulations. The animated version of the figures are provided in a video available at YouTube \url{https://youtu.be/n0C0UZqJmMU}.

\subsection{Attackers Forcing Dilemma}
\noindent
Here we demonstrate the scenario studied in Sec.~\ref{sec:forcing_dilemma}, where the overlap regions are nonempty but the Attackers have a strategy to ensure that the Turret does not reach there.
Figure~\ref{fig:d1d2_simulation} shows how $d_1(t)$ and $d_2(t)$ evolve over time when the Attackers use the strategy introduced in Theorem~\ref{thm:forcing_dilemma}.
\begin{figure}
    \centering
    \includegraphics[width = 0.48\textwidth]{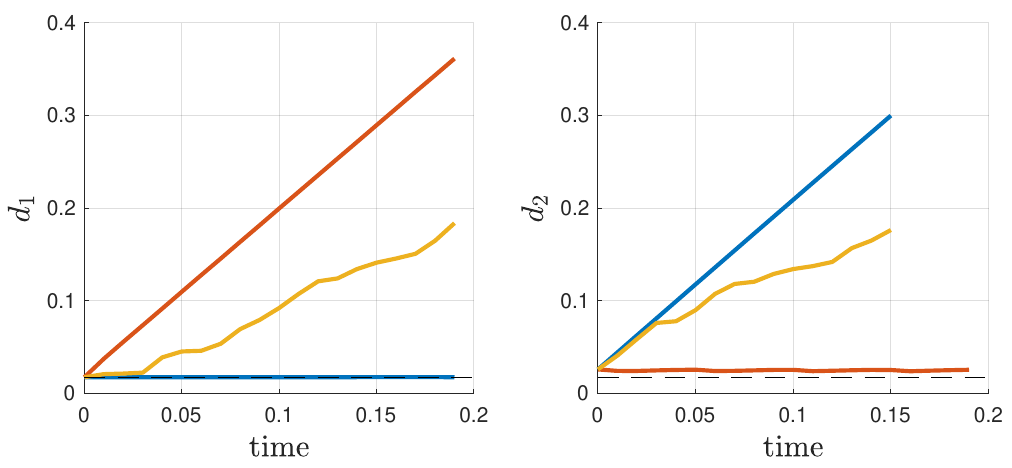}
    \caption{The evolution of $d_1(t)$ and $d_2(t)$ when the Attackers use the strategy described in Theorem~\ref{thm:forcing_dilemma}, and the Turret uses three different strategies: seek the $\reachtwo$ boundary (red), seek the $\reachone$ boundary (blue), and random walk (yellow).}
    \label{fig:d1d2_simulation}
\end{figure}
Three candidate Turret strategies are shown: (i) $\reachone$ seeking, (ii) $\reachtwo$ seeking, and (iii) random walk.
As was proved in the theorem, it can be seen that no Turret strategy results in $d_1(t)$ or $d_2(t)$ reaching a value below $d_\text{min}=\min\{d_1(0),d_2(0)\}$.

\subsection{Initial Conditions}
\noindent
Figure~\ref{fig:numerical_regions} shows how the initial position of $A_2$ changes the existence of dilemma for a given $\theta_T$ and $\pos{A_1}$. 
\begin{figure}
    \centering
    \includegraphics[width=0.47\textwidth]{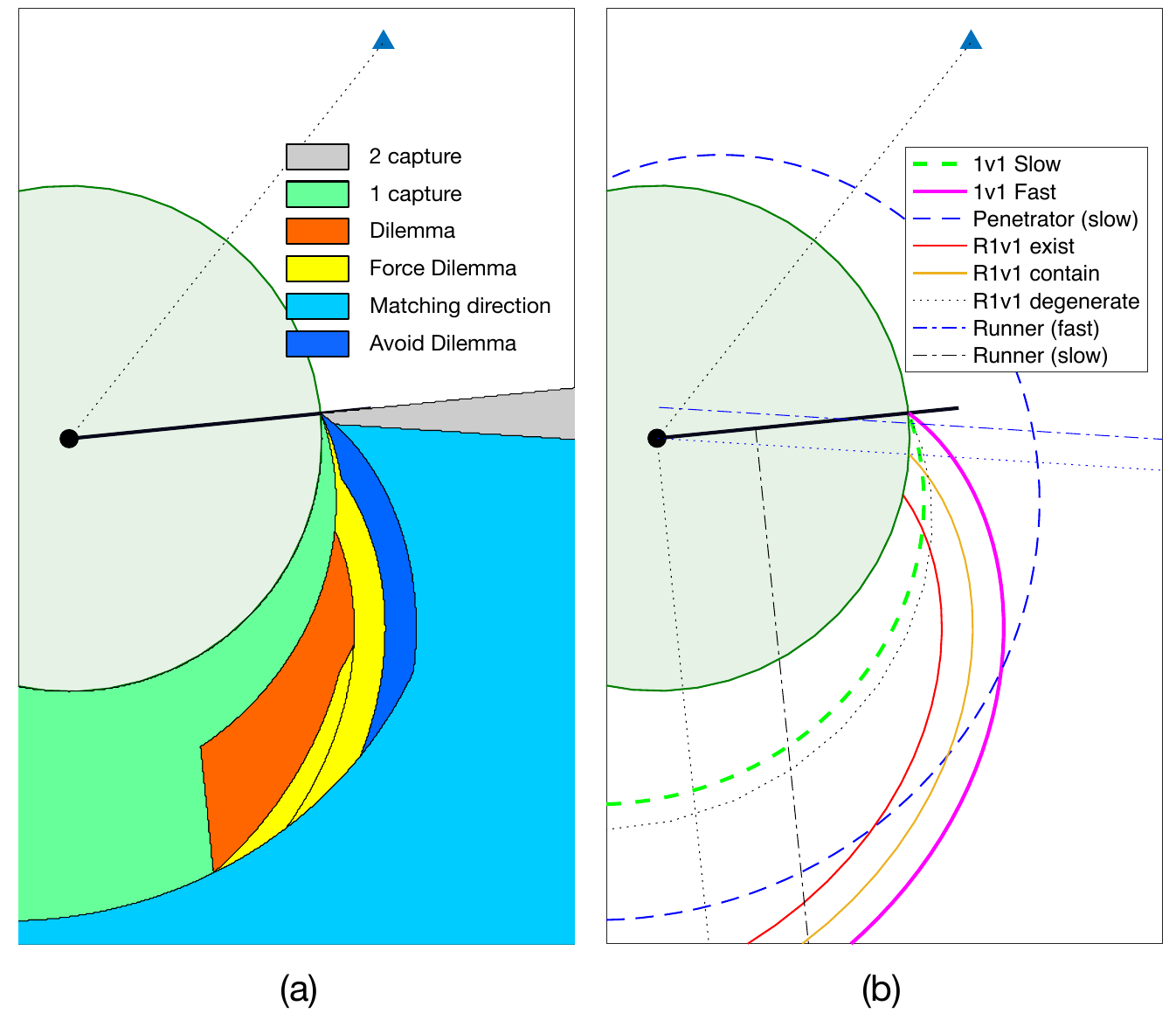}
    \caption{The effect of initial condition on the game cases. The position of $A_2$ is varied for a fixed $\x_{A_1}$ and $\theta_T$. }
    \label{fig:numerical_regions}
\end{figure}
Note that the region above the Turret's pointing angle is omitted since it leads to a degenerate case where CCW motion is optimal against both Attackers.

\paragraph*{Trivial cases} The gray region corresponds to the set of $A_2$ positions for which the Turret has a guarantee to achieve two captures even if the Attackers are fast: i.e., $\theta_T \in \uniontwo(\nufast)$. In the situation shown in the figure, the optimal Turret strategy is to capture $A_2$ first.  The boundary of this region consists of the critical positions for (i) $A_2$ to win the 1v1 fast game, and (ii) $A_2$ to waste enough time as a Runner (the boundary `Runner (fast)' in Fig.~\ref{fig:numerical_regions}b).

The light-green region is another trivial case where 1 and only 1 capture is possible: i.e., $\theta_T \in \unionone(\nufast)\setminus  \uniontwo(\nuslow)$.  Noting that the given $A_1$ position is such that $\theta_T\in \region{A_1}(\nufast)$, the Turret can guarantee at least 1 capture by pursuing $A_1$.  For the Turret to be able to achieve only 1 capture, at least one Attacker must be able to survive even when it is slow. For this situation to arise, the union of the following conditions is necessary and sufficient: (i) $A_2$ can breach alone even when it is slow (the boundary `1v1 Slow' in Fig.~\ref{fig:numerical_regions}b); and (ii) the Attackers can guarantee 1 breach in 2v1 slow game. The latter is equivalent to the intersection of the following two: (ii-a) $A_2$ as a Runner can produce sufficient advantage for $A_1$ to breach if the Turret captures $A_2$ first (the boundary `Runner (slow)' in Fig.~\ref{fig:numerical_regions}b); and (ii-b) $A_2$ as a Penetrator can breach if the Turret captures $A_1$ first (the boundary `Penetrator (slow)' in Fig.~\ref{fig:numerical_regions}b).

\paragraph*{Matching direction}
The light blue region in Fig.~\ref{fig:numerical_regions} depicts the set of $A_2$ positions for which the \emph{matching direction} condition in Lemma~\ref{lem:matching_directions} applies.
This region is the union of the following two cases.
(Note that in either case, the Turret achieves 1 capture against fast Attackers and 2 captures against slow Attackers.)
When $A_2$ is outside of the `Penetrator (slow)' boundary or the `1v1 Fast' boundary, then the turret can move in the appropriate ``matched direction".

\paragraph*{Dilemma Situations}
The interesting dilemma cases arise when $A_2$ is in either of the three regions: `Dilemma', `Force Dilemma', and `Avoid Dilemma' in Fig.~\ref{fig:numerical_regions}a.
The situation described in Sec.~\ref{sec:guaranteed_dilemma} occurs when $A_2$ is in the red region.
This guaranteed dilemma situation changes to a more subtle intermediate configuration discussed in Sec.~\ref{sec:intermediate} when the intersection regions ($\intersectone$ or $\intersecttwo$) are nonempty.
Both in yellow and blue regions in Fig.~\ref{fig:numerical_regions}a, we have either $\intersectone\neq\emptyset$ or $\intersecttwo\neq\emptyset$.

The condition for the existence of $\intersectone$ (and $\reachone$) can be expressed analytically (see \red{Appendix}), and is depicted as `R1v1 exist' in Fig.~\ref{fig:numerical_regions}b.
However, the existence of $\reachtwo$ (lower boundary of 'Dilemma' and 'Force Dilemma' in Fig.~\ref{fig:numerical_regions}a) involves the solution to a transcendental equation and cannot be expressed in closed form.
Finally, the `Force Dilemma' region (yellow) transitions to the `Avoid Dilemma' region (dark blue) when the reachability set contains $\theta_T$. For the combination of $\x_{A_1}$ and $\theta_T$ in this example, we only have $\theta_T \in \reachone$, and this is shown in the blue region.
The boundary between `Force Dilemma' and `Avoid Dilemma' is given by the $A_2$ positions for which the boundary of $\reachone$ coincides with $\theta_T$.
See the \red{Appendix} for a more detailed explanation for this boundary.

While the existence of the above regions as well as their relative sizes vary with $\x_{A_1}$, $\theta_T$, $\nuslow$, and $\nufast$, the example shown here demonstrates that the dilemma situation studied in this paper is not a rare corner case. 


\section{Conclusion}
We formulated a Turret Attacker differential game with uncertainties in the agent speed.
We then proposed a sufficient condition for a strategy to be \emph{information limiting}, and used that to construct an Attacker strategy that forces the Turret to face a dilemma in selecting its strategy.
For a given set of game parameters, we derived the partition of the state space into various cases concerning both the game outcome and also the existence of deception.
Our work serves as an important case study to show how deception can be implemented in a two-sided differential game against an opponent who is not naive to the possibility that he may be being deceived from a position of information disadvantage.

\section*{APPENDIX}
\subsection{Proof of $\frac{\partial \Runnercap}{\partial \theta_T} < -1$}
Noting that
\begin{equation}
    \frac{\partial\theta_{A_i/T}}{\partial \theta_T} = -1,
\end{equation}
we have
\begin{align}
    \frac{\partial\Runnercap}{\partial \theta_T} &= -\frac{\partial\Runnercap}{\partial \theta_{A_i/T}} \\
    &= -\left(\frac{\partial\Runnercap}{\partial \theta_{A_i/T} }- 1\right) \frac{r_{A_i}}{\nu}\cos\left(\Runnercap - \theta_{A_i/T} \right).
\end{align}
Since we have $r_{A_i}\cos\left(\Runnercap - \theta_{A_i/T} \right)>1$ due to the assumption that the Runner can be captured outside of the perimeter, we have
\begin{equation}
    a \triangleq \frac{r_{A_i}}{\nu}\cos\left(\Runnercap - \theta_{A_i/T} \right) > 1. 
\end{equation}
Now the original expression reduces to
\begin{equation}
    \frac{\partial\Runnercap}{\partial \theta_T} = \frac{-a}{a-1} < -1,
\end{equation}
which completes the proof.


\subsection{Expressions for the boundaries}

\subsubsection{Existence of $\reachone$}
Recall that $\reachone$ exists iff $\intersectone$ exists.
For $\intersectone$ to exist, the CW boundary of $\region{A_1}$ must be on the CW side of the CCW boundary of $\region{A_2}$.
The critical case is when those two boundaries coincide, which can be expressed as
\begin{equation}
    \theta_{A_1} - \widthone(r_{A_1};\nufast) = \theta_{A_2} + \widthone(r_{A_2};\nufast)
\end{equation}
where $\widthone(r_{A_i};\nu)$ is the width of $\region{A_i}$ defined in \eqref{eq:widthone}.
The set of critical $A_2$ positions can then be expressed in polar coordinates as follows:
\begin{equation}
    \theta_{A_2} = c_E-F(r_{A_2};\nufast),
\end{equation}
where $c_E=\theta_{A_1}-F(r_{A_1};\nufast)-2F(1;\nufast)$ is a constant.
This gives `R1v1 exist' in Fig.~\ref{fig:numerical_regions}b.

\subsubsection{Avoid Dilemma}
A sufficient condition for the Turret to avoid dilemma is $\theta_T\in\reachone$ according to Theorem~\ref{thm:avoid_dilemma}.
The critical case is when $\theta_T = \theta_{B1}$, where $\theta_{B1}$ is the CCW boundary of $\reachone$ as shown in Fig.~\ref{fig:force_dilemma_overlap}.
The equation for this critical condition switches according to the two cases described in the following, and it causes the `Force Dilemma' region in Fig.~\ref{fig:numerical_regions}a (yellow) to have a non-smooth boundary with the `Avoid Dilemma' (blue) region.

\paragraph*{Nominal case}
The nominal case is when $\intersectone$ is bounded by $\theta_1 \triangleq \theta_{A_1}-\widthone(r_{A_1};\nuslow)$ and $\theta_2 \triangleq \theta_{A_2}-\widthone(r_{A_2};\nuslow)$ as shown in Fig.~\ref{fig:avoidDilemma}.
The boundary of $\reachone$ is given by
\begin{eqnarray}
    \theta_{B1} &=& \cone + \beta\wone,
\end{eqnarray}
where $\beta \triangleq 1/\alpha = \nufast/\nuslow$, $\cone = 0.5(\theta_1+\theta_2)$, $\wone = 0.5(\theta_2-\theta_1)$.
By substitution, the condition $\theta_T = \theta_{B1}$ becomes
\begin{equation}
    \theta_{A_2} = c_n - F(r_{A_1}; \nufast),
\end{equation}
where $c_n = \frac{1}{1+\beta}\left[2\theta_T-(1-\beta)(\theta_{A_1}-\widthone(r_{A_1};\nufast))\right]-F(1;\nufast)$ is a constant.

\paragraph*{Degenerate case}
The degenerate case is when $\region{A_2}$ is so small that it is entirely contained in $\region{A_1}$, which occurs when
\begin{equation}
    \theta_{A_1}-\widthone(r_{A_1};\nufast)\leq \theta_{A_2}-\widthone(r_{A_2};\nufast).
\end{equation}
In this case, we have $\intersectone=\region{A_2}$, and the center and the width of $\intersectone$ is $\cone = \theta_{A_2}$, $\wone = \widthone(r_{A_1};\nufast)$.
The critical condition $\theta_T = \theta_{B1}$ then becomes
\begin{equation}
    \theta_T = \theta_{A_2} + \beta \widthone(r_{A_1};\nufast).
\end{equation}
This is the dotted line `R1v1 degenerate' in Fig.~\ref{fig:numerical_regions}b.



\bibliographystyle{ieeetr}
\bibliography{ref}

\end{document}